%% file: root.tex
\pdfoutput=1
\documentclass[letterpaper, 10pt, conference]{ieeeconf}  

\IEEEoverridecommandlockouts                              

\input{macros}

\title{\LARGE \bf
Verifying Controllers Against Adversarial Examples with Bayesian Optimization
}
\author{Shromona Ghosh$^{1}$, Felix Berkenkamp$^{2}$, Gireeja Ranade$^{3}$, Shaz Qadeer$^{3}$, Ashish Kapoor$^{3}$
\thanks{The primary part of the research for this paper was done  at Microsoft Research, Redmond while Shromona Ghosh and Felix Berkenkamp were interns.}
\thanks{$^{1}$Electrical Engineering and Computer Science Department,
        University of California, Berkeley,
        {\tt\small shromona.ghosh@berkeley.edu}}%
\thanks{$^{2}$Department of Computer Science, ETH Z\"urich, Switzerland
        {\tt\small befelix@inf.ethz.ch}}%
 \thanks{$^{3}$ Microsoft Research, Redmond
 	{\tt\small \{giranade,} {\tt\small qadeer, akapoor\}@microsoft.com }}
}

\usepackage{fancyhdr}
\newcommand{\mytitle}{\textbf{Accepted final version.}
To appear in \textit{2018 IEEE International Conference on Robotics and Automation}.\\
\copyright 2018 IEEE. Personal use of this material is permitted. Permission from IEEE must be obtained for all other uses, in any current or future media, including reprinting/republishing this material for advertising or promotional purposes, creating new collective works, for resale or redistribution to servers or lists, or reuse of any copyrighted component of this work in other works.}
\begin{document}
\maketitle
\thispagestyle{fancy}
\pagestyle{empty}


\input{sections/0-abstract.tex}
\input{sections/1-introduction.tex}
\input{sections/2-prob_statement.tex}
\input{sections/3-modeling.tex}
\input{sections/5-eval.tex}
\input{sections/6-conclusion.tex}
\input{sections/7-acknowledgements.tex}

\bibliographystyle{IEEEtran}
\bibliography{ICRA2018}

\appendix
\input{sections/8-appendix.tex}

\end{document}

%% file: macros.tex
\usepackage{graphicx,caption,subcaption}
\usepackage{amsmath,amssymb,stmaryrd,mathtools}
\usepackage{algorithm,algorithmicx}
\usepackage[noend]{algpseudocode}
\usepackage{acronym}
\usepackage{comment}
\usepackage{color}
\usepackage{cite}
\usepackage{units}
\usepackage{bm}

\usepackage{amsthm}
\usepackage{thmtools}
\usepackage{thm-restate}

\usepackage{graphics} 
\usepackage{times} 

\usepackage[hidelinks]{hyperref}
\usepackage[capitalise]{cleveref}

\DeclareCaptionLabelSeparator{periodspace}{.\quad}
\crefname{equation}{}{} 
\crefname{section}{Sec.}{Sec.}

\newcommand{\envs}{\mathbb{W}}
\newcommand{\env}{\mathbf{w}}
\newcommand{\reals}{\mathbb{R}}

\newcommand{\tree}{\mathcal{T}}

\newcommand{\argmin}{\operatornamewithlimits{argmin}}

\theoremstyle{plain}
\newtheorem{remark}{Remark}
\newtheorem{definition}{Definition}

\newtheorem{theorem}{Theorem}

\newtheorem{example}{Example}
\newtheorem{lemma}{Lemma}
\newtheorem{corollary}{Corollary}


%% file: sections/0-abstract.tex

\begin{abstract}
Recent successes in reinforcement learning have lead to the development of complex controllers for real-world robots. As these robots are deployed in safety-critical applications and interact with humans, it becomes critical to ensure safety in order to avoid causing harm. A first step in this direction is to test the controllers in simulation. To be able to do this, we need to capture what we mean by safety and then efficiently search the space of all behaviors to see if they are safe. In this paper, we present an active-testing framework based on Bayesian Optimization. We specify safety constraints using logic and exploit structure in the problem in order to test the system for adversarial counter examples that violate the safety specifications. These specifications are defined as complex boolean combinations of smooth functions on the trajectories and, unlike reward functions in reinforcement learning, are expressive and impose hard constraints on the system. In our framework, we exploit regularity assumptions on individual functions in form of a Gaussian Process (GP) prior. We combine these into a coherent optimization framework using problem structure. The resulting algorithm is able to provably verify complex safety specifications or alternatively find counter examples. Experimental results show that the proposed method is able to find adversarial examples quickly.
\end{abstract}

%% file: sections/1-introduction.tex

\section{Introduction}

In recent years, research in control theory and robotics has focused on developing efficient controllers for robots that operate in the real world. Controller synthesis techniques such as reinforcement learning, optimal control, and model predictive control have been used to synthesize complex policies.
However, if there is a large amount of uncertainty about the real world environment that the system interacts with, the robustness of the synthesized controller becomes critical. This is particularly true in \textit{safety-critical} systems, where the actions of an autonomous agent may affect human lives. This motivates us to provably verify the properties of controllers in simulation before deployment in the real world.

In this paper, we present an active machine learning framework that is able to verify black-box systems against, or alternatively find, adversarial counter examples to a given set of safety specifications. We test the controller safety under uncertainty that arises from stochastic environments and errors in modeling. In essence, we actively search for adversarial environments under which the controller could have to operate that lead failure modes in simulation.

\begin{figure}
\vspace{2mm} 
\includegraphics{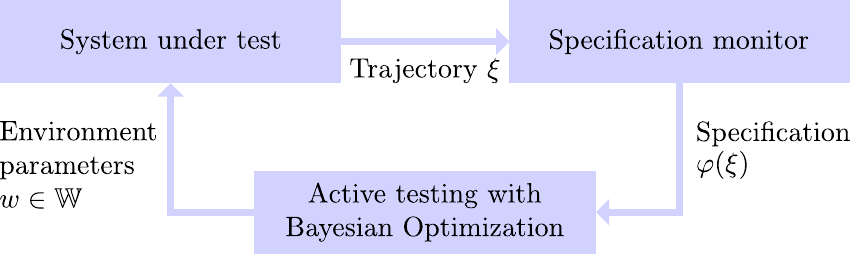}
\caption{Framework for active testing. We query the simulation of the system with environment parameters~$w$ to obtain trajectories~$\xi$. We test these for safety violations in a specification monitor. The Bayesian optimization framework actively queries new parameters~$w$ that are promising candidates to find counter examples that violates the safety specification.}
\label{Fig:framework}
\end{figure}

Historically, designing robust controllers has been considered in control theory \cite{Sastry:1989:ACS:63437,Stengel:1986:SOC:26887}. A common issue with these  techniques is that, although they consider uncertainty, they rely on simple linear models of the underlying system. This means that resulting controllers are often either overly conservative or violate safety constraints if they fail to capture nonlinear effects.

For nonlinear models with complex dynamics, reinforcement learning has been successful for synthesizing high fidelity controllers. Recently, algorithms based on reinforcement learning that can handle uncertainty have been proposed~\cite{KahnVPAL17,niv2002evolution,poupart2008model}, where the performance is measured in expectation.
A fundamental issue with learned controllers is that it is difficult to provide formal guarantees for safety in the presence of uncertainty. For example, a controller for an autonomous vehicle must consider human driver behaviors, pedestrian behaviors, traffic lights, uncertainty due to sensors, etc. Without formally verifying that these controllers are indeed safe, deploying them on the road could lead to loss of property or human lives.

\textit{Formal safety certificates}, i.e., mathematical proofs for safety, have been considered in the formal methods community, where  safety requirements are referred to as a \textit{specification}. There, the goal is to verify that the behaviors of a particular model satisfies a specification~(\cite{Clarke:2000:MC:332656, mitchell2000level}). 
Synthesizing controllers which satisfy a high level temporal specification have been studied in the context of motion planning~\cite{bhatia2010sampling} and for cyber-physical systems~\cite{raman2014model}. 
However, these techniques rely on simple model dynamics. For nonlinear systems, reachability algorithms based on level set methods have been used to approximate backward reachable sets for safety verification~\cite{mitchell2004computing,Berkenkamp2017Safe}. However, these methods suffer from two major drawbacks: (1)~the curse of dimensionality of the state space, which limits them to low-dimensional systems; and (2)~\textit{a priori} knowledge of the system dynamics.

A dual, and often simpler, problem is \textit{falsification}, which
tests the system within a set of environment conditions for adversarial examples.
Adversarial examples have recently been considered for neural networks~\cite{goodfellow2014explaining,papernot2016practical,DBLP:journals/corr/BehzadanM17, carlini2017towards}, where the input is typically perturbed locally in order to find counterexamples. In \cite{DBLP:journals/corr/HuangPGDA17}, the authors compute adversarial perturbations for a trained neural network policy for a subset of white box and black-box systems. However, these local perturbations are often not meaningful for dynamic systems.
Recently,~\cite{dreossi2017compositional, Pei2017} have focused on testing of closed-loop safety critical systems with neural networks by finding ``meaningful" perturbations.

Testing black-box systems in simulators is a well studied problem in the formal methods community~\cite{donze2010breach,c2e2,annpureddy2011s}. The heart of research in black-box testing focuses on developing smarter search techniques which efficiently samples the uncertainty space. Indeed, in recent years, several sequential search algorithms based on heuristics such as Simulated Annealing~\cite{annpureddy2011s}, Tabu search~\cite{deshmukh2015stochastic}, and CMA-ES~\cite{hansen2016cma} have been suggested. Although these algorithms sample the uncertainty space efficiently, they do not utilize any of the information gathered during previous simulations.

One active method that has been used recently for testing black-box systems is Bayesian Optimization (BO)~\cite{mockus2012bayesian}, an optimization method that aims to find the global optimum of an \textit{a priori} unknown function based on noisy evaluations. Typically, BO algorithms are based on Gaussian Process (GP~\cite{3569}) models of the underlying function and certain algorithms provably converge close to the global optimum~\cite{srinivas2009gaussian}. It has been used in robotics to, for example, safely optimize controller parameters of a quadrotor~\cite{Berkenkamp2016SafeOpt}. In the testing setting, BO has been used to actively find counter examples by treating the search problem as a minimization problem in~\cite{DHJMP17} over adversarial control signals. However, the authors do not consider the structure of the problem and thereby violate the smoothness assumptions made by the GP model. As a result, their methods are slow to converge or may fail to find counterexamples.

In this paper, we provide a formal framework that uses BO to actively test and verify closed-loop black-box systems in simulation. We model the relation between environments and safety 
specification using GPs and use BO to predict the environment scenarios most likely to cause failures in our controllers. Unlike previous approaches, we exploit structure in the problem in order to provide a formal way to reason across multiple safety constraints in order to find counterexample. Hence, our approach is able to find counterexamples more quickly than previous approaches. Our main contributions are:
\begin{itemize}
\item An active learning framework for testing and verifying robotic controllers in simulation. Our framework can find adversarial examples for a synthesized controller independent of its structure or how it was synthesized.
\item A common GP framework to model logical safety specifications along with theoretical analysis on when a system is verified.
\end{itemize}

%% file: sections/2-prob_statement.tex

\section{Problem Statement}
\label{sec:problem_statement}

We address the problem of testing complex black-box closed-loop robotic systems in simulation. We assume that we have access to a simulation of the robot that includes the control strategy, i.e., the closed-loop system. The simulator is parameterized by a set of parameters~${\env \in \envs}$, which model all sources of uncertainty. For example, they can represent environment effects such as weather, non-deterministic components such as other agents interacting with the simulator, or uncertain parameters of the physical system, e.g., friction.

The goal is to test whether the system remains safe for all possible sources of uncertainty in~$\envs$. We specify these safety constraints on finite-length trajectories of the system that can be obtained by simulating the robot for a given set of environment parameters~${\env \in \envs}$. Safety constraints on these trajectories are specified using logic. We explain this in detail in~\cref{Sec:Req}, but the result is a specification~$\varphi$ that can, in general, be written as a requirement~$\varphi(\env) > 0, \forall \env \in \envs$. For example,~$\varphi$ can encode state or input constraints that have to be satisfied over time.

We want to test whether there exists an adversarial example~${\env \in \envs}$ for which the specification is violated, i.e.,~${ \varphi(\env) < 0 }$. Typically, adversarial examples are found by randomly sampling the environment and simulating the behaviors. However, this approach does not provide any guarantees and does not allow us to conclude that no adversarial example exist if none are found in our samples. Moreover, since high-fidelity simulations can often be very expensive, we want to minimize the number of simulations that we have to carry out in order to find a counterexample.

We propose an active learning framework for testing, where we utilize the results from previous simulation runs to make more informed decisions about which environment to simulate next. In particular, we pose the search problem for a counterexample as an optimization problem,
\begin{equation}
\label{Eqn:BoolFalse}
\argmin_{\env \in \envs}\, \varphi(\env) ,
\end{equation}
where we want to minimize the number of queries~$\env$ until a counterexample is found or we can verify that no counterexample exists. The main challenge is that the functional dependence~$\varphi(\cdot)$ between parameters in~$\envs$ and the specification is unknown \textit{a priori}, since we treat the simulator as a black-box. Solving this problem is difficult in general, but we can exploit regularity properties of~$\varphi(\env)$. In particular, in the following we use GP to model the specification and use the model to pick parameters that are likely to be counterexamples.

\section{Background}
\label{sec:background}

In this section, we introduce an overview of formal safety specifications and Gaussian processes, which we use in~\cref{sec:main_section} to verify the closed-loop black-box system.

\subsection{Safety Specification}
\label{Sec:Req}

In the formal methods community, complex safety requirements are expressed using automatons~\cite{alur1994theory} and temporal logic~\cite{Pnueli:1977:TLP:1382431.1382534,maler2004monitoring}. These allow us to specify complex constraints, which can also have temporal dependence.

\begin{example}
\label{Ex:quadcopter_spec}
A safety constraint for a quadcopter might be that the quadcopter cannot fly at an altitude~$h$ greater than \unit[3]{m} when the battery level~$b$ is below 30\%."

In logic, we can express this as ``$b < 0.3$ \textbf{implies}($\rightarrow$) $h < 3$'', which in words says if the battery level is less than $30\%$ the quadcopter is flying at a height less than \unit[3]{m}.
\end{example}




Importantly, these kind specifications make no assumptions about the underlying system themselves. They just state requirements that must hold for all simulations in~$\envs$. Formally, a logic specification is a function that tests properties of a particular trajectory. However, we will continue to write~$\varphi(\env)$ to denote the specification that tests trajectories generated by the simulator with parameters $\env$.

A specification~$\varphi$ consists of multiple individual constraints, called predicates, which form the basic building blocks of the logic. These predicates can be combined using a syntax or grammar of logical operations:
\begin{equation}
\varphi := \mu \,|\, \neg \mu \,|\, \varphi \wedge \psi \,|\, \varphi \vee \psi.
\end{equation}
where $\mu:  \Xi \rightarrow \reals$ is a predicate, and is assumed to be a smooth and continuous function of a trajectory $\xi \in \Xi$. The constraint $\mu > 0$ forms the basic building block of the overall system specification $\varphi$. We say a predicate is satisfied if $\mu(\xi)$ is greater than $0$ or falsified otherwise. The operations $\neg, \wedge, \vee$ represent \textit{negation}, \textit{conjunction(and)} and \textit{disjunction(or)} 
, respectively. These basic operations can be combined to define complex boolean formula such as implication, $\rightarrow$, and if-and-only-if, $\leftrightarrow$ using the rules
\begin{equation}
\varphi \rightarrow \psi := \neg \varphi \vee \psi, \text{~and~} 
\varphi \leftrightarrow \psi :=  (\neg \varphi \wedge \neg \psi) \vee (\varphi \wedge \psi).
\label{eq:rewrite_rules}
\end{equation}


Since $\mu$ is a real valued function, we can convert these boolean logic statements into an equivalent equation with continuous output, which defines the \textit{quantitative semantics},
\begin{equation}
\begin{aligned}
\mu(\xi) & := \mu(\xi),
&(\varphi \wedge \psi )(\xi) &:= \min(\varphi(\xi), \psi(\xi)), \\
\neg \mu(\xi) & := -\mu(\xi),
&(\varphi \vee \psi)(\xi)  &:= \max(\varphi(\xi), \psi(\xi)).
\end{aligned}
\label{eq:convert_spec_to_function}
\end{equation}
This allows us to confirm that a logic statement~$\varphi$ holds true for all trajectories generated by simulators~$\envs$, by confirming that the function~$\varphi(\env)$ takes positive values for all~${\env \in \envs}$.

In the quantitative semantics~\cref{eq:convert_spec_to_function}, the satisfaction of a requirement is no longer a yes or no answer, but can be quantified by a real number. The nature of this quantification is similar to that of a reward function, where lower values indicate a larger safety violation. This allows us to introduce a ranking among failures: ${\varphi(\env_1) < \varphi(\env_2)}$ implies $\env_1$ is a more "dangerous" failure case than $\env_2$. 
%
To guarantees safety, we have to take a pessimistic outlook, and denote $\varphi(\env) \leq 0$ as a violation and $\varphi(\env) > 0$ as satisfaction of the specification~$\varphi$.

\begin{example}
Let us look at the specification in~\cref{Ex:quadcopter_spec}, ${\varphi := (b < 0.3 )\rightarrow (h < 3)}$. Applying the re-write rule~\cref{eq:rewrite_rules}, this can be written as ${\neg(b < 0.3) \vee (h < 3)}$. Applying the quantitative semantics~\cref{eq:convert_spec_to_function}, we get ${\varphi = \max(b > 0.3, h < 3)}$, which consists of two predicates, ${\mu_1 = b - 0.3}$ and ${\mu_2 = 3 - h}$. Intuitively, this means $\varphi > 0$, i.e., the specification is satisfied, if the battery is greater than 30$\%$ or if the quadcopter flies at an altitude less than 3m .
\end{example}

\subsection{Gaussian Process}
\label{Sec:GP}

For general black-box systems, the dependence of the specification $\varphi(\cdot)$ on the parameters~$\env \in \envs$ is unknown \textit{a priori}. We use a GP to approximate each predicate $\mu(\cdot)$ in the domain $\envs$. We detail the modeling of $\varphi(\cdot)$ in~\cref{sec:main_section}. The following introduction about GPs is based on~\cite{3569}.

GPs are non-parametric regression method from machine learning, where the goal is to find an approximation of the nonlinear function $\mu : \envs \rightarrow \reals$ from an environment $\env \in \envs$ to the function value $\mu$. This is done by considering the function values $\mu(\env)$ to be random variables, such that any finite number of them have a joint Gaussian distribution.

The Bayesian, non-parametric regression is based on a prior mean function and the kernel function $k(\env, \env')$, which defines the covariance between the function values $\mu(\env), \mu(\env')$ at two points $\env, \env' \in \envs$. We set the prior mean to zero, since we do not have any knowledge about the system. The choice of kernel function is problem-dependent and encodes assumptions about the unknown function.

We can obtain the posterior distribution of a function value $\mu(\env)$ at an arbitrary state $\env \in \envs$ by conditioning the GP distribution of $\mu$ on a set of $n$ past measurements,
${\mathbf{y}_n = (\hat{\mu}(\env_1),\dots,\hat{\mu}(\env_n))}$ at environment scenarios ${W_n = \{\env_1,\dots,\env_n\}}$, where $\hat{\mu}(\env) = \mu(\env) + \omega$ and ${\omega \sim \mathcal{N}(0, \sigma^2) }$ is Gaussian noise. The posterior over $\mu(\env)$ is a GP distribution again, with mean $m_n(\env)$, covariance $k_n(\env,\env′)$, and variance $\sigma_n(\env)$:
\begin{equation}
\begin{split}
\label{Eqn:Predict_GP}
m_n(\env) &= \mathbf{k}_n(\env) (\mathbf{K}_n + \mathbf{I}_n \sigma^2)^{-1} \mathbf{y}_n, \\
k_n(\env, \env') & = k(\env, \env') - \mathbf{k}_n(\env)(\mathbf{K}_n + \mathbf{I}_n \sigma^2)^{-1}\mathbf{k}_n^T(\env'), \\
\sigma_n^2(\env) & = k_n(\env,\env'),
\end{split}
\end{equation}
where the vector $\mathbf{k}_n(\env) =  [k(\env, \env_1), \dots , k(\env, \env_n)]$  contains the covariances between the new environment, $\env$, and the environment scenarios in $W_n$, the kernel matrix ${ \mathbf{K}_n \in \reals^{n\times n} }$ has entries ${[\mathbf{K}_n](i,j) = k(\env_i, \env_j)}$, with ${i,j \in \{1,\dots,n\}}$, and $\mathbf{I}_n \in \reals^{n\times n}$ is the identity matrix.

\subsection{Bayesian Optimization (BO)}
\label{sec:bayesian_optimization}

In the following we use BO in order to find the minimum of the unknown function~$\varphi$, which we construct using the GP models on~$\mu$ in~\cref{sec:main_section}. BO uses a GP model to query parameters that are informative about the minimum of the function. In particular, the \textsc{GP-LCB} algorithm from~\cite{srinivas2009gaussian} uses the GP prediction and associated uncertainty in~\cref{Eqn:Predict_GP} to trade off exploration and exploitation by, at iteration~$n$, selecting an environment according to
\begin{equation}
\label{Eqn:BO_f_acqu}
\env_n = \argmin_{\env \in \envs} m_{n-1}(\env) - \beta_n^{1/2} \sigma_{n-1}(\env) ,
\end{equation}
where $\beta_n$ determines the confidence interval. We provide an appropriate choice for~$\beta_n$ in~\cref{thm:verif}. 

At each iteration,~\cref{Eqn:BO_f_acqu} selects parameters for which the lower confidence bound of the GP is minimal. Repeatedly evaluating the true function~$\varphi$ at samples given by~\cref{Eqn:BO_f_acqu} improves the GP model and decreases uncertainty at candidate locations for the minimum, such that the global minimum is found eventually~\cite{srinivas2009gaussian}.

%% file: sections/3-modeling.tex

\section{Active Testing for Counterexamples}
\label{sec:main_section}

In this section, we show how to model specifications~$\varphi$ in~\cref{Eqn:BoolFalse} using GPs without violating smoothness assumptions and use this to find adversarial counterexamples.

In order to use BO to optimize~\cref{Eqn:BoolFalse}, we need to construct reliable confidence intervals on~$\varphi$. However, if we were to model $\varphi$ as a GP with commonly-used kernels, it would need it to be a smooth function of~$\env$. Even though the predicates,~$\mu$, are typically smooth functions of the trajectories, and hence smooth in $\env$, conjunction and disjunction ($\min$ and $\max$) in~\cref{eq:convert_spec_to_function} are non-smooth operators that render $\varphi$ to become non-smooth as well.
Instead, we exploit the structure of the specification~$\varphi$ and decompose~$\varphi$ into a parse tree, where the leaf nodes are the predicates.

\begin{definition}[Parse Tree $\tree$]
Given a specification formula $\varphi$, the corresponding parse tree,~$\tree$, has leaf nodes that correspond to function predicates, while other nodes are $\max$ (disjunctions) and $\min$ (conjunctions).
\end{definition}
A parse tree is an equivalent graphical representation of~$\varphi$. For example, consider the specification
\begin{equation}
\label{Eqn:Example_parseTree}
\varphi := (\mu_1 \vee \mu_2) \rightarrow (\mu_3 \vee \mu_4) = (\neg \mu_1 \wedge \neg \mu_2) \vee (\mu_3 \vee \mu_4),
\end{equation}
where the second equality follows from De-Morgan's law.
We can obtain an equivalent function~${ \varphi(\env) }$ with~\cref{eq:convert_spec_to_function},
\begin{equation}
\begin{aligned}
\varphi(\env) = \max\big( &\min (-\mu_1(\env),\, -\mu_2(\env)), \\ & \max (\mu_3(\env),\, \mu_4(\env)) \big) .
\end{aligned}
\label{eq:parse_tree_function}
\end{equation}
The parse tree, $\tree$, for $\varphi$ in~\cref{eq:parse_tree_function} is shown in~\cref{Fig:parse_tree}. We can use the parse tree to decompose any complex specification into~$\min$ and~$\max$ functions of the individual predicates; that is, $\varphi(\env) = \tree(\mu_1(\env), \dots, \mu_q(\env))$.

\begin{figure}[t]
\centering
\includegraphics[scale=0.45]{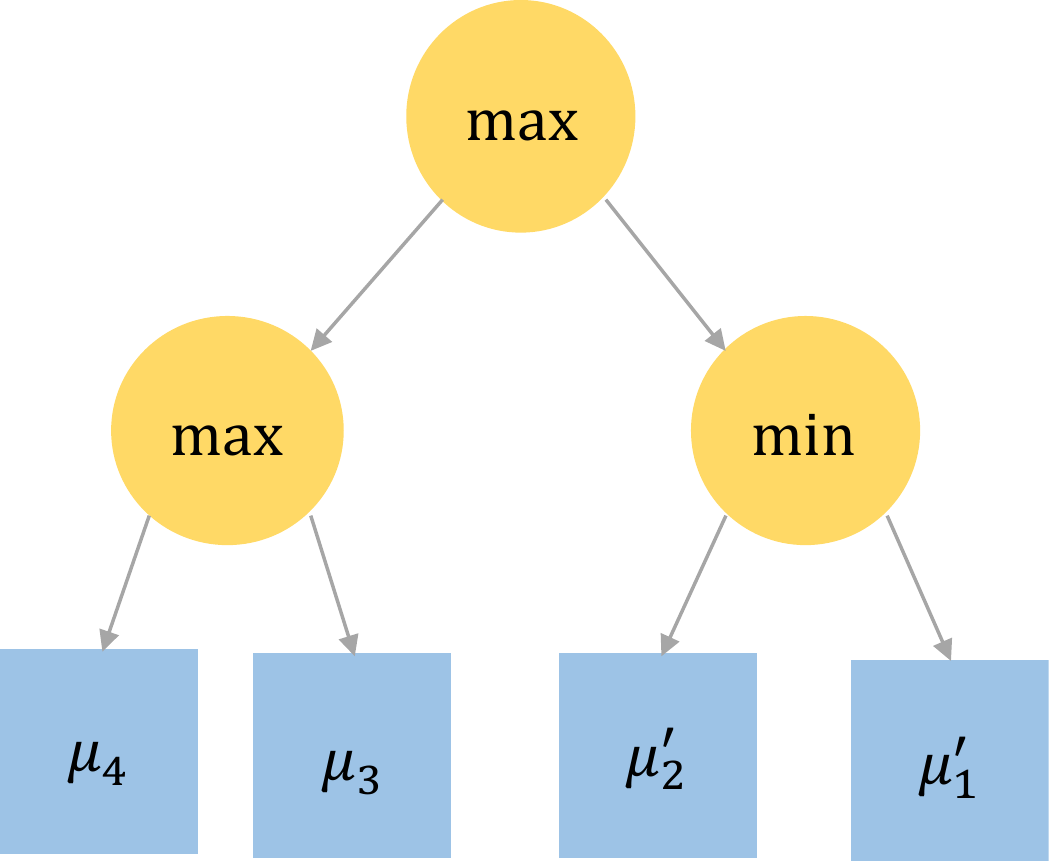}
\caption{Equivalent parse tree~$\tree$ for~$\varphi$ in~\cref{Eqn:Example_parseTree} to the function~\cref{eq:parse_tree_function}. We replace the predicates~$\mu_i$ with their corresponding pessimistic GP predictions to obtain a lower bound on~$\varphi(\env)$.}
\label{Fig:parse_tree}
\end{figure}


We now model each predicate~$\mu_i(w)$ in the parse tree $\tree$ of $\varphi$ with a GP and combine them with the parse tree to obtain confidence intervals on the overall specification~$\varphi(\env)$ for BO.
GP-LCB as expressed in~\cref{Eqn:BO_f_acqu} can be used to search for the minimum for a single GP. A key insight to extending~\cref{Eqn:BO_f_acqu} across multiple GPs, is that the minimum of~\cref{Eqn:BoolFalse} is, with high probability, lower bounded by the lower-confidence interval of one of the GPs used to model the predicates of $\varphi$. This is because, the $\max$ and $\min$ operators do not change the value of the predicates, but only make a choice between them. As a consequence, we can model the smooth parts of $\varphi$, i.e., the predicates, using GPs and then consider the non-smoothness through the parse tree.

For each predicate~$\mu_i$ in the parse tree $\tree$ of $\varphi$, we construct a lower confidence bound~$l_i = m^i_{n-1}(\env) - \beta_n^{1/2} \sigma^i_{n-1}(\env)$, where $m^i, \sigma^i$ are the mean and standard deviation of the GP corresponding to $\mu_i$. From this, we can construct a lower-confidence interval on~$\varphi$ as
$\tree(l_1(\env), \dots, l_q(\env))$, where we replace the $i$th leaf node~$\mu_i$ of the parse tree with the pessimistic prediction $l_i$ of the corresponding GP. Similar to~\cref{Eqn:BO_f_acqu}, the corresponding acquisition function for BO uses this lower bound to select the next evaluation point,
\begin{equation}
\label{Eqn:BO_multi_GP}
\env_n = \argmin_{\env \in \envs}\, \tree(l_1(\env), \dots, l_q(\env)).
\end{equation}
Intuitively, the next environment selected to simulate is the one that minimizes the worst-case predictions on~$\varphi$. Effectively, we propagate the confidence intervals associated with the GP for each predicates through the parse tree~$\tree$ in order to obtain predictions about~$\varphi$ directly. Note, that~\cref{Eqn:BO_multi_GP} does not return an environment sample that minimizes the satisfaction of all the predicates, it only minimizes the lower bound on~$\varphi$.

\cref{Algo:active_testing} describes our active testing procedure.
The algorithm proceeds by first computing the parse tree $\tree$ from the specification, $\varphi$. At each iteration~$n$ of BO, we select new environment parameters~$\env_n$ according to~\cref{Eqn:BO_multi_GP}. We then simulate the system with parameters~$\env_n$ and evaluate each predicate~$\mu_i$ on the simulated trajectories. Lastly, we update each GP with the corresponding measurement of~$\mu_i$. The algorithm either returns a counterexample that minimizes~\cref{Eqn:BoolFalse}; or when~$\tree(l_1(\env), \dots, l_q(\env))$ is greater then zero, and we can conclude that the system has been verified.

\subsection{Theoretical Results}
\label{sec:theory}

We can transfer theoretical convergence results for GP-LCB~\cite{srinivas2009gaussian} to the setting of~\cref{Algo:active_testing}. To do this, we need to make structural assumptions about the predicates. In particular, we assume that they have bounded norm in the Reproducing Kernel Hilbert Space (RKHS, \cite{steinwart2008support}) that corresponds to the GP's kernel. These are well-behaved functions of the form~$\mu_i(\env) = \sum_{j=0} \alpha_j k_i(\env, \env_j)$ with representer points~$\env_j$ and weights~$\alpha$ that decay sufficiently quickly. We leverage theoretical results from~\cite{ChowdhuryG17} and~\cite{Berkenkamp2016SafeOpt} that allow us to build reliable confidence intervals using the GP models from~\cref{Sec:GP}. We have the following result.
\begin{restatable}{theorem}{maintheorem}
Assume that each predicate $\mu_i$ has RKHS norm bounded by~$B_i$ and that the measurement noise is $\sigma$-sub-Gaussian. Select $\delta \in (0,1)$, $\env_n$ according to~\cref{Eqn:BO_multi_GP}, and let ${\beta^{1/2}_n = \sum_i B_i + 4 \sigma \sqrt{ 1 + \ln(1/\delta) + \sum_i I(\mathbf{y}^i_{n-1}; \mu_i)}}$. If~$\tree(l_1(\env_n), \dots, l_q(\env_n)) > 0$, then with probability at least $1-\delta$ we have that $\min_{\env \in \envs} \varphi(\env) > 0$  and the system has been verified against all environments in~$\envs$.
\label{thm:verif}
\end{restatable}
Here~$I(\mathbf{y}^i_{n-1}; \mu_i)$ is the mutual information between~$\mathbf{y}^i_{n-1}$, the~$n-1$ noisy measurements of~$\mu_i$, and the GP prior of~$\mu_i$. This function was shown to be sublinear in~$n$ for many commonly-used kernels in~\cite{srinivas2009gaussian}, see the appendix for more details.
\cref{thm:verif} states that we can verify the system against adversarial examples with high probability, by checking whether the worst-case lower-confidence bound is greater than zero. We provide additional theoretical results about the existence of a finite~$n$ such that the system can be verified up to~$\epsilon$ accuracy in the appendix.

\begin{algorithm}[t]
\caption{Active Testing with Bayesian Optimization}\label{Algo:active_testing}
\begin{algorithmic}[1]
\Procedure{ActiveTesting}{$\varphi, \envs, \beta, \textit{GP}s$}
\State Build parse tree $\tree$ based on specification $\varphi$
\For {$n = 0,\dots$} \Comment{Until budget or convergence}
\State $l_i(\env) = \mu_i(\env) - \beta_n^{1/2} \sigma_i(\env), \, i=1,\dots,q$
\State $w_n = \argmin_{w \in \envs} \tree(l_1(\env), \dots, l_q(\env))$
\State Update each GP model of the predicates with \par
\hspace{1.5ex} measurements $(\env_n,\, \mu_i(\env_n))$.
\EndFor
\State return~$\min_i\, \varphi(w_i)$, the worst result.
\EndProcedure
\end{algorithmic}
\end{algorithm}

%% file: sections/5-eval.tex

\section{Evaluation}
\label{sec:evaluation}

In this section, we evaluate our method on several challenging test cases. A Python implementation of our framework and the following experiments can be found at \mbox{\url{https://github.com/shromonag/adversarial_testing.git}}

In order to use~\cref{Algo:active_testing}, we have to solve the optimization problem~\cref{Eqn:BO_multi_GP}. In practice, different optimization techniques have been proposed to find the global minimum of the function. One popular algorithm is DIRECT~\cite{finkel2003direct}, a gradient-free optimization method. An alternative is to use gradient-based methods together with random-restarts. Particularly, we sample a large number of potential environment scenarios at random from $\envs$, and run seperate optimization routines to minimize~\cref{Eqn:BO_multi_GP} from these.

Another challenge is that the dimensionality of the optimization problem can often be very large. However, methods that allow for more efficient computation do exist. These methods reduce the effective size of the input space and thereby make the optimization problem more tractable. One possibility is to use random embedding to reduce the input dimension as done in Random Embedding Bayesian Optimization (REMBO~\cite{wang2013bayesian}). We can then model the GP in this smaller input dimension and carry out BO in the lower dimension input space.

\subsection{Modeling smooth functions vs non-smooth function}
In the following, we show the effectiveness of modeling smooth functions by GPs and considering the non-smooth operations in the BO search as opposed to modeling the non-smooth function by a single GP.

\label{Eval:smooth_non_smooth}
Consider the following, illustrative optimization problem,
\begin{equation}
\label{Eqn:sin_cos}
w^{*} = \argmin_{w \in (0, 10)} \max (\sin(w)+ 0.65, \cos(w)+0.65)
\end{equation}

\begin{figure*}
\centering
\begin{subfigure}[t]{0.32\textwidth}
\includegraphics[scale=0.6]{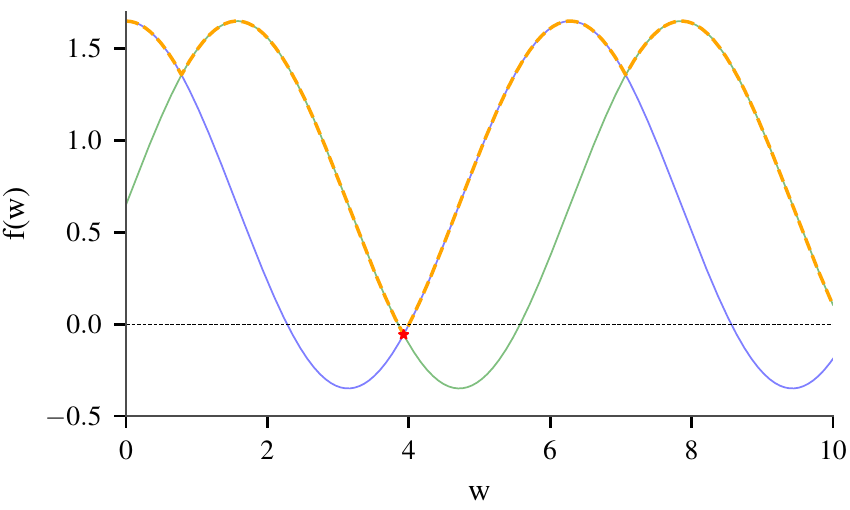}
\caption{True functions.}
\label{Fig:sin_cos}
\end{subfigure}
\hfill
\begin{subfigure}[t]{0.32\textwidth}
\includegraphics[scale=0.6]{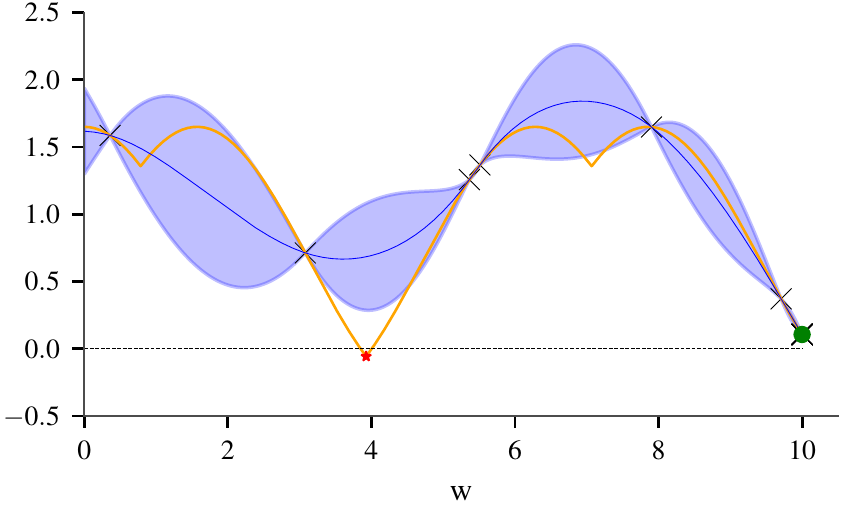}
\caption{Modeling the non-smooth~$\varphi(w)$ directly.}
\label{Fig:max_x}
\end{subfigure}
\hfill
\begin{subfigure}[t]{0.32\textwidth}
\includegraphics[scale=0.6]{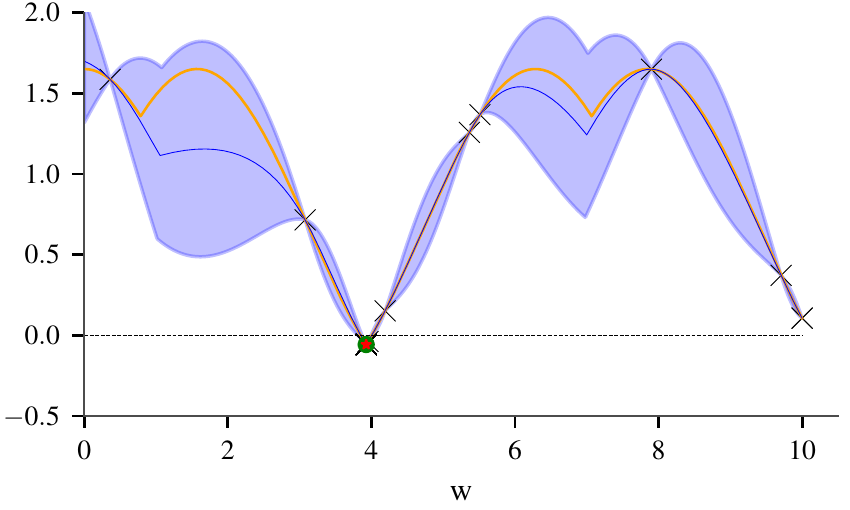}
\caption{Modeling and combining smooth predicates~$\mu$.}
\label{Fig:sin_cos_x}
\end{subfigure}
\caption{The dashed orange line in~\cref{Fig:sin_cos} represents the true, non-smooth optimization function in~\cref{Eqn:sin_cos} while the green and blue line represent $\sin(w)$ and $\cos(w)$ respectively. Modeling this function directly as a GP leads to model errors~\cref{Fig:max_x}, where the $95\%$ confidence interval of the GP (blue shaded) with  mean estimate (in blue line) does not capture the true function~$\varphi(w)$ in orange. In fact, the minimum (red star) is not contained within the shaded region, causing the optimization to diverge. BO converges to the green dot, where $\varphi(w)> 0$ which is not a counterexample. Instead, modeling the two predicates individually and combining them with the parse tree, leads to the model in~\cref{Fig:sin_cos_x}. Here, the true function is completely captured in the confidence interval. As a consequence, BO converges to the global minimum (the red star and green dot converge).}
\label{Fig:model}
\end{figure*}

We consider two modeling scenarios, one where we model $\max(\sin(w), \cos(w))$ as a single GP, and another where we model $\sin(w)$ by one GP and $\cos(w)$ by another.
We initialize the GP models for $\sin(w)$, $\cos(w)$ and $\max(\sin(w), \cos(w))$ with 5 samples chosen in random. We then use BO to find $w^{*}$. We were  able to model smooth functions like $\sin(w)$ and $\cos(w)$ with GPs, even with fewer samples. At each iteration of BO, we computed the next sample by solving for the $w \in (0, 10)$ which minimized the maximum across the two GPs. This quickly stabilizes to the true $w^*$~(\cref{Fig:sin_cos_x}).
When we model $\max(\sin(w), \cos(w))$ using a GP,  in~\cref{Fig:max_x}, the initial 5 samples were not able to model it well. In fact, the original function in orange is not contained within the uncertainty bounds of the GP.  Hence, in each iteration of BO, where we chose $w \in (0,10)$ which minimized this function, we were never able to converge $w^{*}$.  It is not surprising to see that, given these models, BO does not always converge when we model non-smooth functions such as in~\cref{Eqn:sin_cos}.

To support our claim, we repeat this experiment 15 times with different initial samples. In each  experiment we run BO for 50 iterations. When modeling $\sin(w)$ and $\cos(w)$ as separate GPs, BO stabilized to $w^{*}$ in about 5 iterations in all 15 experiments. However, when modeling $\max(\sin(w), \cos(w))$ as a single GP, it takes over 35 iterations to converge and in 5 out of the 15 cases, it did not converge to $w^{*}$. We show these two different behaviors in~\cref{Fig:stabilizing}.
\begin{figure*}
\centering
\begin{subfigure}[t]{0.45\linewidth}
\centering
\includegraphics[scale=0.75]{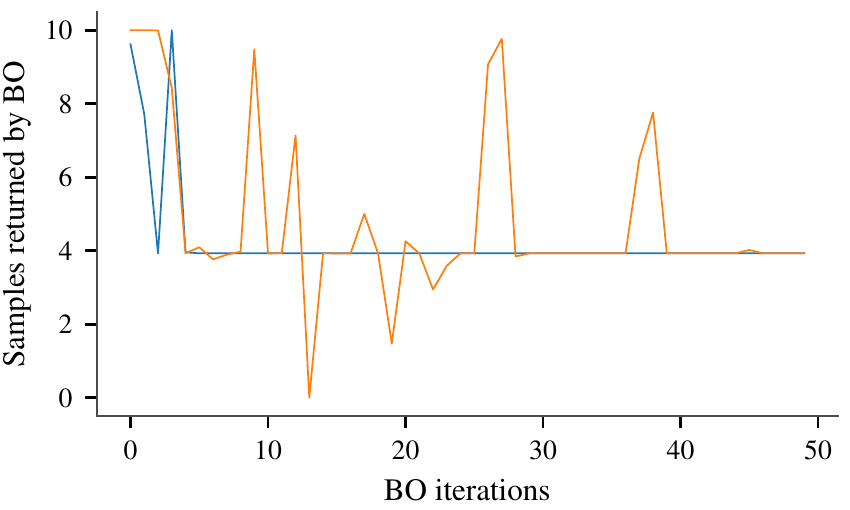}
\caption{Modeling as separate GPs take around 5 iterations to stabilize to $w^{*}$~(in blue), while modeling as a single GP takes around 45 iteration to stabilize to $w^{*}$~(in orange)}
\label{Fig:stab_late}
\end{subfigure}
\hfill
\begin{subfigure}[t]{0.45\linewidth}
\centering
\includegraphics[scale=0.75]{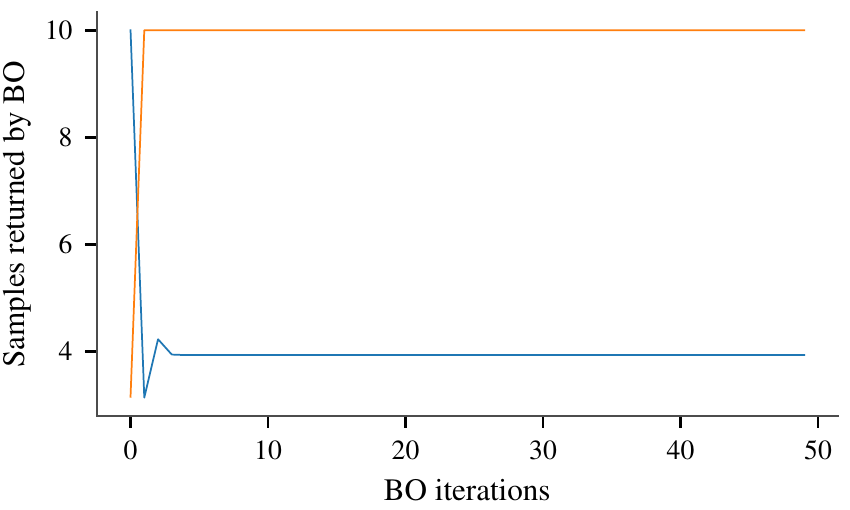}
\caption{Modeling as separate GPs take around 5 iterations to stabilize to $w^{*}$~(in blue), while modeling as a single GP does not stabilize to $w^{*}$~(in orange)}
\label{Fig:fail_joint}
\end{subfigure}
\caption{The orange and blue lines in~\cref{Fig:stab_late} and ~\cref{Fig:fail_joint} show the evolution of samples returned over the BO iterations when~\cref{Eqn:sin_cos} is modeled as a single GP and multiple GPs respectively for two different initialization. We see the that when modeling as a single GP, it takes longer to stabilize to $w^{*}$ and in some cases~(\cref{Fig:fail_joint}) does not stabilize to $w^{*}$.}
\label{Fig:stabilizing}
\end{figure*}

\subsection{Collision Avoidance with High Dimensional Uncertainty}
\label{Sec:car_high}
Consider an autonomous car that travels on a straight road with a obstacle at $x_{obs}$. We require that the car can come to a stop before colliding with an obstacle. The car has two states; location, $x$, and velocity, $v$; and one control input acceleration; $a$. The dynamics of the car is given by,
\begin{equation}
\dot{x} = v , \quad
\dot{v}  = a.
\end{equation}

Our safety specification for collision avoidance is given by, $\varphi = \min(x_{obs} - x(t))$, i.e., the minimum distance between the position of the car and the obstacle over a horizon of length $100$. We assume that the car does not know where the obstacle is \textit{a priori}, but receives locations of the obstacle through a sensor at each time instant, $x_s(t)$. The controller is a simple linear state feedback control, $K$, such that at time $t$, $a(t) = K \cdot \left[ \begin{matrix} x(t)-x_s(t), \,v(t) \end{matrix} \right]^\mathrm{T}$.

We assume that the car initially starts at location $x_{init} = 0$, with a velocity $v_{init} = \unit[3]{m/s}$. Let the obstacle be at $x_{obs} = 5$, which is not known by the car. Instead, it receives sensor readings for the location of the obstacle such that $x_s = [4.5, 5.5]$.  If $\varphi$ is negative, then $x(t) > x_{obs}$ for some $t$ which signifies collision. Moreover, we constrain the acceleration to lie in $a \in [-3, 3]$.

The domain of our uncertainty is $\envs = [4.5, 5.5]^{100}$, i.e., the sensor readings $x_s$ over the horizon $H = 100$.
We compare across three experimental setups, first, we model the GP in the original space of $\envs$ i.e., with $100$ inputs; second, we model the GP in a lower dimension input space as described in the preamble of this section; and third, we randomly sample inputs and test them. We run BO for 250 iterations on the GPs, and consider 250 random samples for the random testing. We repeat this experiment 10 times and show our results in~\cref{Fig:car_example}.
\begin{figure}[t]
\centering
\includegraphics[scale=0.75]{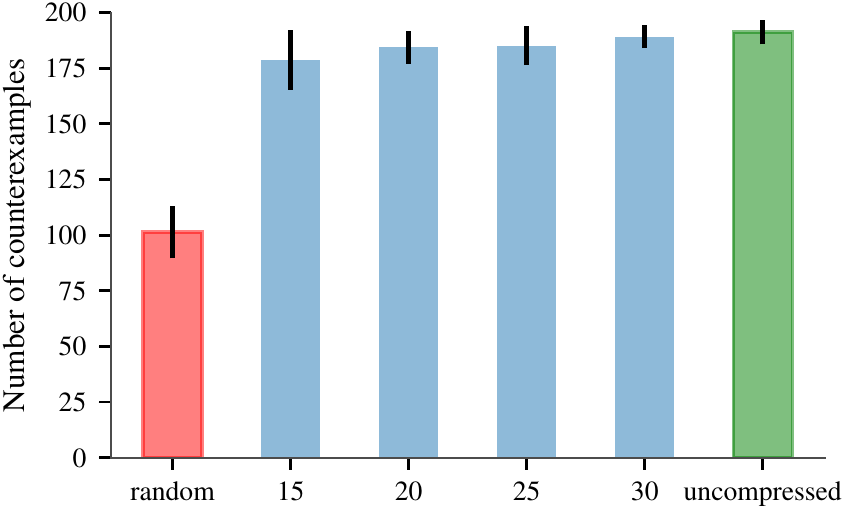}
\caption{The red, blue and green bars shows the average number of counterexamples found using random sampling; applying BO on the reduced input space and original input space respectively for the example in~\cref{Sec:car_high}.  The black lines show the standard deviation across the experiments.}
\label{Fig:car_example}
\end{figure}
The green and blue bar in~\cref{Fig:car_example} show the  average number of counterexamples returned running BO on the GP defined over the original input space and in the low dimension input space. In general, active testing in the high-dimensional input space gives the best results, which deteriorates with an increase in compression of the input space. Random testing, shown in red performs the worst. This is not surprising as, (1) $250$ samples is not sufficient to cover an input space of $100$ dimensions uniformly; and (2) the samples are all independent of each other.
Moreover, in the uncompressed input case, the specification evaluated at the worst counterexample, $\varphi(w^{*})$, has a mean and standard deviation of  $-0.0138$  and $0.004$ as compared to $-0.0067$ and $0.0011$ for random sampling.

\subsection{OpenAI Gym Environments}
We interfaced our tool with environments from OpenAI gym~\cite{1606.01540} to test controllers from Open AI baselines~\cite{baselines}. For brevity, we refer the details of the environments to~\cite{ai_envs}.
In both case studies, we introduce uncertainty around the parameters the controller has been trained for. The rationale behind this is that the parameters in a simulator are an estimate of the true values. This ensures that counterexamples found, can indeed occur in the real system.

\subsubsection{Reacher}
\label{Sec:reach}
In the reacher environment, we have a 2D robot trying to reach a target. For this environment we have six sources of uncertainty: two for the goal position, $(x_{goal}, y_{goal}) \in [-0.2, 0.2]^2$, two for state perturbations $(\delta_x, \delta_y) \in [-0.1, 0.1]^2$ and two for velocity perturbations $(\delta_{vx}, \delta_{vy}) \in [-0.005, 0.005]^2$.
The state of the reacher is tuple with the current location, $\textbf{x} =(x, y)$, velocity $\textbf{v} = (v_x, v_y)$, and rotation, $\theta$. A trajectory of the system, $\xi$, is a sequence of states over time, i.e., $\xi = (\textbf{x}(t), \textbf{v}(t), \theta(t)), t = 0, 1, 2, \dots$.
Our uncertainty space is, $\envs =  [-0.2, 0.2]^2 \times [-0.1, 0.1]^2 \times  [-0.005, 0.005]^2$.
Given an instance of $w \in \envs$, the trajectory, $\xi$, of the system is uniquely defined.

We trained a controller using the Proximal Policy Optimization (PPO)~\cite{SchulmanWDRK17} implementation available at Open AI baselines. We determine a trajectory to be safe if either the reacher reaches the goal, or if it does not rotate unnecessarily.
This can be captured as $\varphi = \mu_1 \vee \mu_2$, where,
$\mu_1(w)$ is the minimum distance between the trajectory and the goal position, and $\mu_2$ is total rotation accumulated over the trajectory; and its continuous variant, $\varphi = \max(\mu_1, \mu_2)$.

Using our modeling approach, we model this using two GPs, one for $\mu_1$ and another for $\mu_2$. We compare this to modeling $\varphi$ as a single GP and random sampling. We run 200 BO iterations and consider 200 random samples for random testing. We repeat this experiment 10 times.
\begin{figure}
\centering
\includegraphics[scale=0.75]{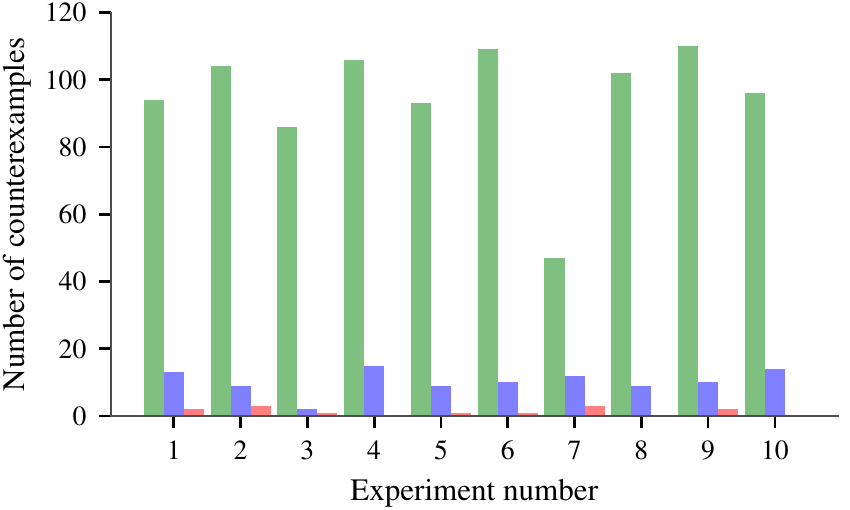}
\caption{The green, blue and red bars show the number of counter examples generated when modeling $\mu_1, \mu_2$ as separate GPs; modeling $\varphi$ as a single GP and random testing respectively for the reacher example~(\cref{Sec:reach}). Our modeling paradigm, finds more counterexamples compared to the other two methods.}
\label{Fig:reacher_exps}
\end{figure}
In~\cref{Fig:reacher_exps}, we plot the number of counterexamples found by each of the three methods over 10 runs of the experiment.  Modeling the predicates by separate GPs and applying BO across them~(shown in green) consistently performs better than applying BO on a single GP modeling $\varphi$~(shown in blue) and random testing (shown in red). We see the that random testing performs very poorly, in some cases (experiment runs $4, 8, 10$) finds no counterexamples.

By modeling the predicates separately, the specification evaluated at the worst counterexample, $\varphi(w^{*})$, has a mean and standard deviation of  $-0.1283$  and $0.0006$ as compared to $-0.1212$ and $0.0042$ when considering a single GP. This suggests, that using our modeling paradigm BO converges (since the standard deviation is small) to a more falsifying counterexample (since the mean is smaller).

\subsubsection{Mountain Car Environment}
\label{Sec:mc}
The mountain car environment in OpenAI gym, is a car on a one-dimensional track, positioned between two mountains. The goal is to drive the car up the mountain on the right. The environment comes with one source of uncertainty, the initial state $x_{init} \in [-0.6, -0.4]$.
We introduced four other sources of uncertainty, for the initial velocity, $v_{init} \in [-0.025, 0.025]$; goal location, $x_{goal} \in [0.4, 0.6]$; maximum speed, $v_{max} \in [0.55, 0.75]$ and maximum power magnitude, $p_{max} \in [0.0005, 0.0025]$.
The state of the mountain car is a tuple with the current location, $x$, and velocity, $v$. A trajectory of the system, $\xi$, is a sequence of states over time, i.e., $\xi = (x(t), v(t)), t = 0, 1, 2, \dots$.
Our uncertainty space is given by, $\envs = [-0.6, -0.4] \times [-0.025, 0.025] \times  [0.4, 0.6] \times [0.55, 0.75] \times [0.0005, 0.0025]$.  Given an instance of $w \in \envs$, the trajectory, $\xi$, of the system is uniquely defined.

We trained two controllers one using PPO and another using an actor critic method (DDPG) for continuous Deep Q-learning~\cite{LillicrapHPHETS15}. We determine a trajectory to be safe, if it reaches the goal quickly or if does not deviate too much from its initial location and always maintains its velocity in some bound. Our safety specification can be written as $\varphi = \mu_1 \vee (\mu_2 \wedge \mu_3)$, where,
$\mu_1(w)$ is time taken to reach the goal, $\mu_2$ is the deviation from the initial location and $\mu_3$ is the deviation from the  velocity bound; and its continuous variant of $\varphi = \max(\mu_1, \min(\mu_2, \mu_3))$.
We model $\varphi$, by modeling each predicate, $\mu$, by a GP. We compare this to modeling $\varphi$ with a single GP and random sampling.
We run 200 BO iterations for the GPs and consider 200 random samples for random testing. We repeat this experiment 10 times. We show our results in~\cref{Fig:mc}, where we plot the number of counterexamples found by each of the three methods over 10 runs of the experiment for each controller.
~\cref{Fig:mc} demonstrates the strength of our approach. The number of counterexamples found by our method (in green bar) is much higher compared to random sampling (in red) and modeling $\varphi$ as a single GP (in blue). In~\cref{Fig:mc_ppo} the blue bars are smaller than even the ones in red, suggesting random sampling performs better than applying BO on the GP modeling $\varphi$. The is because the GP is not able to model $\varphi$, and is so far away from the true model, that the sample returned by the BO is worse than if were to sample randomly.

\begin{figure*}[t]
\centering
\begin{subfigure}[t]{0.45\textwidth}
\centering
\includegraphics[scale=0.75]{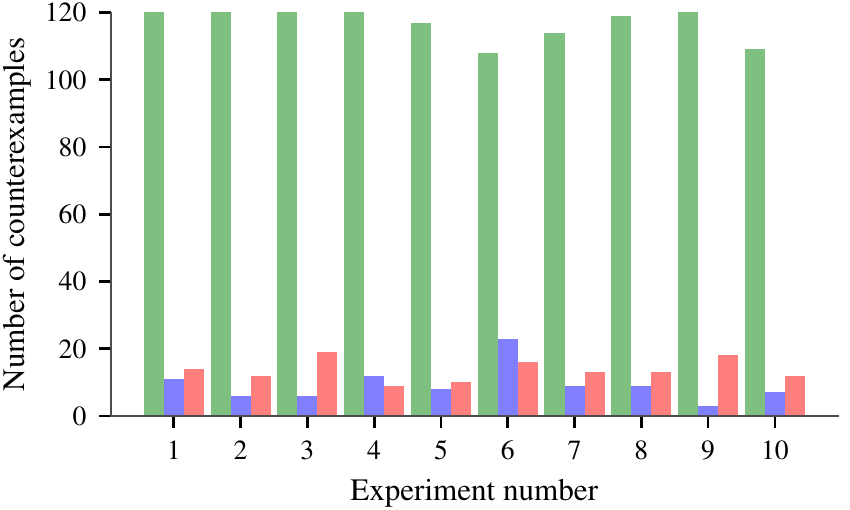}
\caption{Controller trained with PPO}
\label{Fig:mc_ppo}
\end{subfigure}
\begin{subfigure}[t]{0.45\textwidth}
\centering
\includegraphics[scale=0.75]{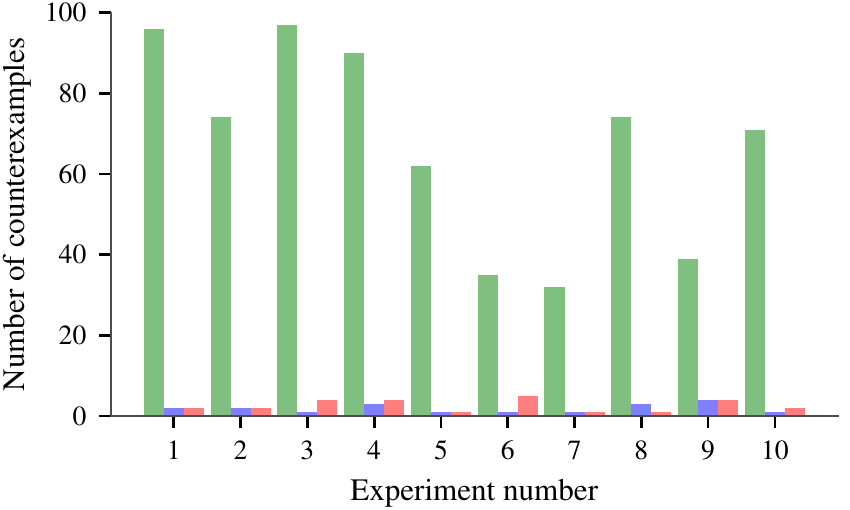}
\caption{Controller trained with DDPG}
\label{Fig:mc_ddpg}
\end{subfigure}
\caption{The green, blue and red bars show the number of counter examples generated when modeling $\mu_1, \mu_2$ as separate GPs, modeling $\varphi$ as a GP and random testing respectively for the mountain car example~(\cref{Sec:mc}). While our modeling paradigm, finds orders of magnitude more counterexample compared to the other two methods, we notice that modeling $\varphi$ as a single GP performs much worse than random sampling for the controller trained with PPO~\cref{Fig:mc_ppo} and comparable for the controller trained with DDPG~\cref{Fig:mc_ddpg}.}
\label{Fig:mc}
\end{figure*}

This is further highlighted by the value of the specification at worst counterexample, $\varphi(w^*)$. The mean and standard deviation for $\varphi(w^*)$ over the 10 experiment runs is $-0.5435$ and $0.028$ for our method, $-0.3902$ and $0.0621$ when $\varphi$ is modeled as a single GP; and $-0.04379$ and $0.0596$ for random sampling. A similar but less drastic result holds in the case of the controller trained with DDPG.

%% file: sections/6-conclusion.tex

\section{Conclusion}
We presented an \textit{active testing} framework that uses Bayesian Optimization to test and verify closed-loop robotic systems in simulation. Our framework handles complex logic specifications and models them efficiently using Gaussian Processes in order to find adversarial examples faster. We showed the effectiveness of our framework on controllers designed on OpenAI gym environments. As future work, we would like to extend this framework to test more complex robotic systems and find regions in the environment parameter space where the closed-loop control is expected to fail.

%% file: sections/7-acknowledgements.tex

\section*{Acknowledgments}
Research reported in this paper was sponsored by the Army Research Laboratory and was accomplished under Cooperative Agreement Number W911NF-17-2-0196 \footnote{The views and conclusions contained in this document are those of the authors and should not be interpreted as representing the oﬃcial policies, either expressed or implied, of the Army Research Laboratory or the U.S. Government. The U.S. Government is authorized to reproduce and distribute reprints for Government purposes notwithstanding any copyright notation here on.} and was accomplished under Cooperative Agreement Number W911NF-17-2-0196; and in part by Toyota under the iCyPhy center.

%% file: sections/8-appendix.tex

\section{Proofs}

In this section, we prove the convergence of our algorithm under specified regularity assumptions on the underlying predicates.
Consider the specification
\begin{equation}
    \varphi(\env) = \tree( \mu_1(\env), \dots, \mu_{q}(\env)),
\end{equation}
where $q$ represents the number of predicates. Let the domain of the predicate indices be represented by, $\mathcal{I}= \{1, \dots, q\}$. The convergence proofs for classical Bayesian optimization in~\cite{srinivas2009gaussian,ChowdhuryG17} proceed by building reliable confidence intervals for the underlying function and then showing, that these confidence intervals concentrate quickly enough at the location of the optimum under the proposed evaluation strategy. For ease of exposition, we assume that measurements of each predicate~$\mu_i$ are corrupted by the same measurement noise.

To leverage these proofs, we need to account for the fact that our GP model is composed of several individual predicates and that we obtain one measurement for each predicates at every iteration of the algorithm.

We start by defining a composite function ${f:\envs \times \mathcal{I} \rightarrow \reals}$, which returns the function values for the individual predicates indexed by~$i$.
\begin{equation}
f(\env, i) = \mu_i(\env)
\end{equation}
The function $f(\cdot, \cdot)$ is a single output function, which can be modeled with a single GP with a scalar output over the extended input space, $\envs \times \mathcal{I}$. For example, if we assume that the predicates are independent of each other, the kernel function for $f$ would look like,
\begin{equation}
\begin{split}
    k((\env, i), (\env', i')) = \begin{cases}
    k_i(\env, \env') \text{ if } i = i' \\
    0 \text{ otherwise}
    \end{cases}
\end{split},
\label{eq:independent_kernel}
\end{equation}
where $k_i$ is the kernel function corresponding to the GP for the $i-$th predicate, $\mu_i$. It is straightforward to include correlations between functions in this formulation too.

This reformulation allows us to build reliable confidence intervals on the underlying predicates, given regularity assumptions. In particular, we make the assumption that the function~$f$ has bounded norm in the Reproducing Kernel Hilbert Space (RKHS, \cite{steinwart2008support}) corresponding to the same kernel~$k(\cdot, \cdot)$ that is used for the GP on~$f$.

\begin{remark}
Note, that this model is more general then the case where we assume that each predicate, $\mu_i$, individually has bounded RKHS norm $B_i$. In this case, the function, $f(\env, i)$ has RKHS norm with respect to the kernel in~\cref{eq:independent_kernel} bounded by~${ B = \sum_i^{q} B_i }$.
\label{note:sum_of_bounded_rkhs}
\end{remark}

\begin{lemma}
Assume that $f$ has RKHS norm bounded by $B$ and the measurements are corrupted by $\sigma$-sub-Gaussian noise. If $\beta_{n \cdot q}^{1/2} = B + 4\sigma \sqrt{I(y_{q \cdot (n - 1)}; f) + 1 + \ln{1/\delta}}$, then the following holds for all environment scenarios, $\env \in \envs$, predicate indices, $i \in \mathcal{I}$, and iterations $n \geq 1$ jointly with probability at least $1- \delta$,
\begin{equation}
    |f(\env, i) - m^i_{q \cdot(n-1)}(\env, i)|\leq \beta_{q \cdot n}^{1/2} \sigma^i_{q \cdot (n -1)}(\env, i)
\end{equation}
\label{lem:mu_confidence}
\end{lemma}
\begin{proof}
This follows directly from~\cite{Berkenkamp2016SafeOpt}, which extends the results from~\cite{ChowdhuryG17} and Lemma 5.1 from~\cite{srinivas2009gaussian} to the case of multiple measurements.
\end{proof}

The scaling factor for the confidence intervals,~$\beta_{n\cdot q}$, depends on the mutual information~$I(\mathbf{y}_{q \cdot (n-1)}; f)$ between the GP model of~$f$ and the~$q$ measurements of the individual predicates that we have obtained for each time step so far. It can easily be computed as
\begin{equation}
\begin{aligned}
    I(\mathbf{y}_{q \cdot (n-1)}; f) &= \log ( 1 + \frac{1}{\sigma^2} \mathbf{K}_{q \cdot (n-1)} ), \\
    &= \sum_{j=1}^{n - 1} \sum_{i=1}^q \log (1 + \sigma^2_{j \cdot q}(\env_j, i) / \sigma^2),
\end{aligned}
\label{eq:mutual_information}
\end{equation}
where~$\mathbf{K}_{q \cdot (n-1)}$ is the kernel matrix of the single GP over the extended parameter space and the inner sum in the second equation indicates the fact that we obtain~$q$ measurements at every iteration.

Based on these individual confidence intervals on~$\mu$, we can construct confidence intervals on~$\varphi$. In particular, let
\begin{equation}
    \begin{aligned}
    l_i(\env) &= m_{q\cdot(n - 1)}(\env, i) - \beta^{1/2}_{q \cdot n} \sigma_{q \cdot (n - 1)}(\env, i) \\
    u_i(\env) &= m_{q \cdot (n - 1)}(\env, i) + \beta^{1/2}_{q \cdot n} \sigma_{q \cdot (n - 1)}(\env, i)
    \end{aligned}
\end{equation}
be the lower and upper confidence intervals on each predicate. From this, we construct reliable confidence intervals on~$\varphi(\env)$ as follows:

\begin{lemma}
Under the assumptions of~\cref{lem:mu_confidence}. Let~$\tree$ be the parse tree corresponding to~$\varphi$. Then the following holds for all environment scenarios, $\env \in \envs$ and iterations $n \geq 1$ jointly with probability at least $1- \delta$,
\begin{equation}
    \tree(l_1(\env), \dots, l_q(\env)) \leq \varphi(\env) \leq \tree(u_1(\env), \dots, u_q(\env))
\end{equation}
\label{lem:confidence_on_combined_fun}
\end{lemma}
\begin{proof}
This is a direct consequence of~\cref{lem:mu_confidence} and the properties of the~$\min$ and~$\max$ operators.
\end{proof}

We are now able to prove the main theorem as a direct consequence of~\cref{lem:confidence_on_combined_fun}.

\maintheorem*
\begin{proof}
For independent variables the mutual information decomposes additively and following~\cref{note:sum_of_bounded_rkhs} this is a direct consequence of~\cref{lem:confidence_on_combined_fun}, since~${ \tree(l_1(\env), \dots, l_q(\env)) \leq \varphi(\env) }$ holds for all ${ \env \in \envs }$ with probability at least~$1 - \delta$.
\end{proof}

\subsection{Convergence proof}

In the following, we prove a stronger result about convergence of our algorithm.

The key quantity in the behavior of the algorithm is the mutual information~\cref{eq:mutual_information}. Importantly, it was shown in~\cite{Berkenkamp2016SafeOpt} that it can be upper bounded by the worst-case mutual information, the information capacity, which in turn was shown to be sublinear by~\cite{srinivas2009gaussian}. In particular, let~$\mathbf{f}_\mathbb{W}$ denote the noisy measurements obtained when evaluating the function~$f$ at points in~$\mathbb{W}$. The mutual information obtained by the algorithm can be bounded according to
\begin{equation}
    \begin{split}
        I(\mathbf{f}_{\envs_n \times \mathcal{I}}; f)
        & \leq \max_{\bar{\envs} \subset \envs, |\bar{\envs}| \leq n}I(\mathbf{f}_{\bar{\envs}\times \mathcal{I}}; f);\\
        & \leq \max_{\mathcal{D} \subset \envs \times \mathcal{I}, |\mathcal{D}| \leq n \cdot q}  I (\mathbf{f}_{\mathcal{D}}; f); \\
        &=\gamma_{q\cdot n} ,
    \end{split}
    \label{Eqn:MI_bound}
\end{equation}
where~$\gamma_n$ is the worst-case mutual information that we can obtain from~$n$ measurements,
\begin{equation}
    \gamma_n = \max_{\mathcal{D} \subset \envs \times \mathcal{I}, |\mathcal{D}| = n} I(\mathbf{f}_{\mathcal{D}}; f).
\end{equation}
This quantity was shown to be sublinear in~$n$ for many commonly-used kernels in~\cite{srinivas2009gaussian}.

A key quantity to show convergence of the algorithm is the instantaneous regret,
\begin{equation}
    r_n = \min_{\env \in \envs} \varphi(\env) - \varphi(\env_n),
\end{equation}
the difference between the unknown true minimizer of~$\varphi$ and the environment parameters~$\env_n$ that~\cref{Algo:active_testing} selects at iteration~$n$. If the instantaneous regret is equal to zero, the algorithm has converged.

In the following, we will show that the cumulative regret, ${R_n = \sum_{i=1}^{n} r_i}$ is sublinear in~$n$, which implies convergence of~\cref{Algo:active_testing}.

We start by bounding the regret in terms of the confidence intervals on~$\mu_i$.
\begin{lemma}
Fix $n \geq 1$, if $|f(\env, i) - m_{q \cdot (n-1)}(\env, i)|\leq \beta^{1/2}_{q\cdot n}\sigma_{q \cdot(n-1)}(\env, i)$ for all $\env, i \in \envs \times \mathcal{I}$, then the regret is bounded by $r_n \leq 2 \beta^{1/2}_{q\cdot n} \max_i \sigma_{q \cdot (n-1)}(\env, i)$.
\label{lem:2}
\end{lemma}
\begin{proof}
The proof is analogous to~\cite[Lemma 5.2]{srinivas2009gaussian}. The maximum standard deviation follows from the properties of the $\max$ and $\min$ operators in the parse tree~$\tree$. In particular, let~$a_1, b_1, a_2, b_2 \in \mathbb{R}$ with~$a_1 - b_1 < a_2 - b_2$. Then for all~$c_1 \in [-b_1, b_1]$ and~$c_2 \in [-b_2, b_2]$ we have that
\begin{equation}
    a_1 - b_1 \leq \min(a_1 + c_1, a_2 + c_2) \leq a_1 + b_1.
\end{equation}
The~$\max$ operator is analogous. Thus, since the parse tree~$\tree$ is composed only of min and max nodes, the regret is bounded by the maximum error over all predicates. The result follows.
\end{proof}

\begin{lemma}
Pick $\delta \in (0,1)$ and $\beta_{q \cdot n}$ as shown in~\cref{lem:mu_confidence}, then the following holds with probability at least $1-\delta$,
\begin{equation}
    \sum_{i=1}^n r_n^2 \leq \beta_{q \cdot n} C_1 q \mathbf{I}(\mathbf{f}_{\envs_n \times \mathcal{I}}; f) \leq \beta_{q \cdot n} C_1 \gamma_{q\cdot n}
\end{equation}
where $r_n$ is the regret between the true minimizing environment scenario, $\env^*$ and the current sample, $\env_n$; and $C_1 = 8 / \log{1+\sigma^{-2}}$
\label{lem:3}
\end{lemma}
\begin{proof}
The first inequality follows similar to~\cite[Lemma 5.4]{srinivas2009gaussian} and the proofs in~\cite{Berkenkamp2016SafeOpt}. In particular, as in \cite{Berkenkamp2016SafeOpt},
\begin{equation*}
    r_n^2  \leq 4 \beta_{q \cdot n}^2 \max_{i \in \mathcal{I}} \sigma^2_{q \cdot (n-1)}(\env_n, i)
\end{equation*}
The second inequality follows from~\cref{Eqn:MI_bound}.
\end{proof}

\begin{lemma}
Under the assumptions of~\cref{lem:confidence_on_combined_fun}, let $\delta \in (0,1)$ and choose $\env_n$ according to~\cref{Eqn:BO_multi_GP}. Then, the cumulative regret $R_N$ over $N$ iterations of~\cref{Algo:active_testing} is bounded with high probability,
\begin{equation}
    \text{Pr} \left\{ R_n \leq \sqrt{C_1 \beta_N N \gamma_{q \cdot N}} \quad \forall N \geq 1 \right\} \geq 1 - \delta
\end{equation}
where $C_1 = \frac{8}{\log{1+\sigma^{-2}}}$.
\label{lem:4}
\end{lemma}
\begin{proof}
Since, $R_N = \sum_{i=1}^{N} r_i$, from Cauchy-Schwartz inequality we have, $R_N^2 \leq N \sum_{i=1}^N r_i^2$. The rest follows from~\cref{lem:3}.
\end{proof}

We introduce some notation, let
\begin{equation}
    \hat{\env}_n  =\text{argmin}_{\env \in \{\env_1, \dots, \env_n\}} \varphi(\env)
    \label{Eqn:min_env}
\end{equation}
be the minimizing environment scenario sampled by BO in $n$ iterations and let
\begin{equation}
  \env^* = \argmin_{\env \in \envs} \varphi(\env)
\end{equation}
be the unknown, optimal parameter.

\begin{corollary}
For any $\delta \in (0,1)$ and $\epsilon \in \reals^+$, there exits a $n^*$,
\begin{equation}
    \frac{n^*}{\beta_{q \cdot n^*} \gamma_{q \cdot n^*}} = \frac{C_1}{\epsilon^2}
\end{equation}
such that $\forall n \geq n^*$, $\varphi(\env^*) - \varphi(\hat{\env}_n) \leq \epsilon$ holds with probability at least $1 - \delta$.
\label{cor:1}
\end{corollary}
\begin{proof}
The cumulative reward over $n$ iterations, $R_n = \sum_{i=1}^{n} \varphi(\env^*) - \varphi(\env_i)$ where $\env_i$ is the $i$-th BO sample.
Defining $\hat{\env}_n$ as in~\cref{Eqn:min_env} we have,
\begin{equation}
    \begin{split}
        R_n &= \sum_{i=1}^{n} \varphi(\env^*) - \varphi(\env_i)  \\
        & \geq \sum_{i=1}^n \varphi(\env^*) - \varphi(\hat{\env}_n)  \\
        & = n (\varphi(\hat{\env}_n) - \varphi(\env^*))
    \end{split}
\end{equation}
Combining this result with~\cref{lem:4}, we have with probability greater than $1-\delta$ that
\begin{equation}
    \begin{split}
        \varphi(\env^*) - \varphi(\hat{\env}_n)
        &\leq \frac{R_n}{n} \\
        & \leq \sqrt{\frac{C_1 \beta_{q\cdot n} \gamma_{q \cdot n}}{n}}
    \end{split}
\end{equation}
To find, $n^*$, we bound the RHS by $\epsilon$,
\begin{equation}
    \sqrt{\frac{C_1 \beta_{q \cdot n^*} \gamma_{q \cdot n^*}}{n^*}} \leq \epsilon \Rightarrow \frac{n^*}{\beta_{q \cdot n^*} \gamma_{q \cdot n^*}} \geq \frac{C_1}{\epsilon^2}
\end{equation}
For $n > n^*$, the minimum $\varphi(\hat{\env}_n) \leq \varphi(\hat{\env}_{n^*}) \implies \varphi(\hat{\env}_n)- \varphi(\env^*) \leq \epsilon$.
\end{proof}

We are now ready to prove our main convergence theorem.

\begin{theorem}
Under the assumptions of~\cref{lem:confidence_on_combined_fun}, choose $\delta \in (0,1)$, $\epsilon \in \reals^+$ and define $n^*$ using~\cref{cor:1}. If $n \geq n^*$ and $\varphi(\hat{\env}_n) > \epsilon$, then, with probability greater than $1-\delta$, the following statements hold jointly
\begin{itemize}
    \item $\varphi(\env^*) > 0$
    \item The closed loop satisfies $\varphi$, i.e., the control can safely control the system in all environment scenarios, $\envs$
    \item The system has been verified against all environments, $\envs$
\end{itemize}
\end{theorem}
\begin{proof}
This holds from~\cref{lem:4} and~\cref{cor:1}. From~\cref{cor:1}, we have $\forall n \geq n^*$, $\text{Pr}(\varphi(\env^*) - \varphi(\hat{\env}_n) \leq \epsilon) > 1-\delta$. If $\exists n \geq n^*$, such that $\varphi(\hat{\env}_n) > \epsilon$, then we have $\text{Pr}(\varphi(\env^*) > 0) 1-\delta$, i.e., the minimum value $\varphi$ can achieve on the closed loop system is greater than $0$. $\varphi$ is hence, satisfied by our system in all $\env \in \envs$.
\end{proof}